\documentclass[journal]{IEEEtran}

\ifCLASSINFOpdf

\else

\fi

\usepackage{graphics}

\usepackage{amsmath}
\usepackage{amsthm}

\usepackage{amssymb}  
\usepackage{amsfonts} 
\usepackage{graphicx,float,wrapfig}
\usepackage{cite}
\usepackage{bm}
\usepackage{siunitx}
\usepackage{tikz}
\usepackage[compatibility=false]{caption}
\usepackage{eucal}

\usepackage{array, tabularx, makecell, booktabs}
\usepackage{soul,color}
\usepackage[linesnumbered,ruled]{algorithm2e}
\usepackage{url}
\graphicspath{{figures/}}
\usepackage{subcaption}
\usepackage{dcolumn}
\usepackage[english]{babel}
\usepackage[utf8]{inputenc}
\usepackage{accents}

\newtheorem{theorem}{Theorem}

\newtheorem{remark}{Remark}

\newcommand{\norm}[1]{\left\lVert#1\right\rVert}
\newcommand\munderbar[1]{%
  \underaccent{\bar}{#1}}
 \makeatother
\begin{document}
%
\title{Compact Cooperative Adaptive Cruise Control for Energy Saving: Air Drag Modelling and Simulation}
%
%
%

\author{Yeojun~Kim,~\IEEEmembership{Member,~IEEE,}, 
        Jacopo~Guanetti,~\IEEEmembership{Member,~IEEE,}
        and~Francesco~Borrelli,~\IEEEmembership{Fellow,~IEEE}
\thanks{Y. Kim and F. Borrelli are with the Department
of Mechanical Engineering, University of California at Berkeley, Berkeley, CA 94701, USA.}
\thanks{J. Guanetti is with AVConnect Inc., Berkeley, CA 94704, USA.}
\thanks{Corresponding author: Yeojun Kim {\small yk4938@berkeley.edu}}
}

\markboth{Journal of \LaTeX\ Class Files,~Vol.~14, No.~8, August~2015}%
{Shell \MakeLowercase{\textit{et al.}}: Bare Demo of IEEEtran.cls for IEEE Journals}

\maketitle

\begin{abstract}
This paper studies the value of communicated motion predictions in the longitudinal control of connected automated vehicles (CAVs).
We focus on a safe cooperative adaptive cruise control (CACC) design and analyze the value of vehicle-to-vehicle (V2V) communication in the presence of uncertain front vehicle acceleration.
The interest in CACC is motivated by the potential improvement in energy consumption and road throughput.
In order to quantify this potential, we characterize experimentally the relationship between inter-vehicular gap, vehicle speed, and (reduction of) energy consumption for a compact plug-in hybrid electric vehicle.
The resulting model is leveraged to show efficacy of our control design, which pursues small inter-vehicle gaps between consecutive CAVs and, therefore, improved energy efficiency.
Our proposed control design is based on a robust model predictive control framework to systematically account for the system uncertainties.
We present a set of thorough simulations aimed at quantifying energy efficiency improvement when vehicle states and predictions exchanged via V2V communication are used in the control law.
\end{abstract}


%
\IEEEpeerreviewmaketitle

\section{Introduction}
%
%
%
%
\IEEEPARstart{V}{ehicle} connectivity enables exchange of information between vehicles, infrastructures, remote data, and computing centers.
Connected automated vehicles (CAVs) are regarded as the principal beneficiaries of such augmented stream of information.
CAVs can explicitly cooperate among themselves by means of vehicle-to-vehicle (V2V) communication. 
In the past decades, several cooperative vehicle applications have been studied and implemented, both by academia and industry, to enhance safety and performance of the transportation sector. 
Examples include the California PATH program in the 1990s, which demonstrated the cooperative driving of a group (or \textit{platoon}) of connected vehicles on a highway to increase traffic throughput \cite{shladover1991automated, shladover2007path}, the Grand Cooperative Driving Challenges in the Netherlands in 2011 and 2016, which demonstrated platoon control and other interactions among vehicles in different traffic situations \cite{kianfar2012design, englund2016grand}.

Cooperative Adaptive Cruise Control (CACC) is an instance of automated longitudinal motion control for CAVs, or, in other words, an enhancement of adaptive cruise control (ACC) enabled by V2V communication.
Previous research on CACC has focused on the \emph{string stability} of vehicle stream to maximize the road throughput \cite{swaroop1996string,dunbar2012distributed, dolk2017event}.
String stability is defined as the uniform boundedness of all the states of the interconnected system at all times, if the initial states of the interconnected system are uniformly bounded \cite{swaroop1996string}.
In platoon control, string stability implies that the leading vehicle's velocity and position perturbation does not cause a larger amplification in the following vehicle's velocity and position \cite{zhou2005range}.


Another major motivation of CACC development is to minimize energy consumption and pollutant emissions of vehicles. 
This can be achieved by different strategies. 
Examples include platoon coordination \cite{wang2017developing}, driving at more energy efficient speed \cite{barth2008real}, minimizing idling time and braking on arterial corridors with traffic lights \cite{yang2016eco}, and optimizing acceleration/deceleration with the gear shift \cite{shao2017robust}.

Other works on CACC which pursue energy saving are focused on the minimization of the inter-vehicular distance, both to maximize the road throughput and to reduce the vehicle aerodynamic drag (and consequently its energy consumption) \cite{tsugawa2011automated,andersson2012online,lammert2014effect}.
The relationship between inter-vehicular distance, aerodynamic drag, and vehicle energy consumption is often cited in the CACC literature; however, its experimental validation is limited.
The effect on energy consumption can be substantial in heavy-duty vehicles, which have large frontal areas.
An experimental characterization of such relationship is found in \cite{hucho2013aerodynamics} for a heavy-duty vehicle.
Passenger vehicles can also benefit from a reduced inter-vehicular distance: in \cite{zabat1995aerodynamic}, one-eighth scale models of minivan vehicles were tested in a wind tunnel, and it was shown that, even for such vehicles, a reduced inter-vehicular distance translates directly to reduced aerodynamic drag and energy consumption.
It is reported that the vehicles driving at very short gaps can reduce their consumption up to $30\%$ on highways and up to $10\%$ on urban roads.
However, to the best of the authors' knowledge, the literature lacks an experimental characterization of the aerodynamic drag and the vehicle energy consumption as a function of the inter-vehicular distance, \textit{on full scale passenger vehicles}.

Several CACC control designs have been proposed in the literature. 
Examples include classical PID control \cite{naus2010string}, robust $H_{\infty}$ control \cite{kayacan2017multiobjective}, and data-based control such as reinforcement learning \cite{kreidieh2018dissipating}.
Model predictive control (MPC) is also appealing for CACC because its formulation can naturally leverage predictive information offered by V2V communication, such as velocity forecast and/or vehicle model parameters. 
Robust MPC is particularly appropriate for this application due to its capability to provide safety guarantee against distance, velocity, and actuator constraints, in the face of uncertain front vehicle future movements.

\subsection{Contributions}
This paper is focused on the value of V2V communication in longitudinal motion control of CAVs.
In particular, we propose a provably safe CACC which exploits a V2V message, which includes the current and predicted states, and show the controller performance in reducing the inter-vehicular distance and the energy consumption.
The contributions of this work are as follows:
\begin{itemize}
\item We characterize experimentally the aerodynamic drag and vehicle energy consumption as a function of vehicle speed and inter-vehicular gap, for a compact passenger vehicle.
\item We present, through simulation study, a provably safe CACC based on model predictive control and investigate how leveraging V2V communication can lead to the energy saving which is quantified by using the experimental performance curves obtained at the previous point.
\end{itemize}




\subsection{Notation}
The notation $x(k| t)$ indicates the predicted value of variable $x$ at time step $k$ using the available information at time $t$.
The mark above the variable distinguishes the type of that variable. For instance, the true value of variable $x$ is denoted by $\bar{x}$, its measurement is denoted by $\hat{x}$, and $\Tilde{x}(k| t)$ is the estimated state at time $k$ using the available information at time $t$.
The positive part of the real number $x$ is denoted by $x^+$.
$\norm{x}_Q$ denotes the 2 norm of a vector $Qx \in \mathbb R^{n_x}$, or $(x'Qx)^{\frac{1}{2}}$.
For $x\in \mathbb{R}$, $(x)^+$ operator is defined as
\begin{equation}
    (x)^+ := \begin{cases}
    x & \textnormal{if } x\geq 0 ,\\
    0 & \textnormal{else. }
    \end{cases}
\end{equation}
\section{Vehicle dynamics modeling}
\label{sec:veh_dyn_model}
We model the dynamics of the longitudinal vehicle speed as \cite{Guzzella2013}:
\begin{equation}
\label{eq:long_dyn}
\begin{aligned}
M \frac{d}{dt}v(t) = F_w(t) - &M g (\sin{\vartheta(t)}  + C_r \cos{\vartheta(t)}) \\
&- C_v v(t) - \frac{1}{2} \rho A C_x v(t)^2 ,
\end{aligned}
\end{equation}
where 
$F_w$ is the drivetrain force at the wheels (both for traction and braking),
$M$ is the vehicle mass,
$\vartheta$ is the road slope,
$A$ is the front area,
$\rho$ is the air density,
$C_r$ is the rolling coefficient,
$C_v$ is the viscous friction coefficient,
$C_x$ is the air drag coefficient.
We also denote the wheel radius as $R_w$, the wheel speed as $\omega_w = v/R_w$, the wheel torque as $T_w = F_w R_w$.

Motivated by the lack of data on the effects of inter-vehicular distance on $C_x$, we designed and executed a set of experiments aimed at filling this knowledge gap.

\subsection{Experimental setup}
\begin{figure}
\centering
\includegraphics[width=\linewidth,height=\textheight,keepaspectratio]{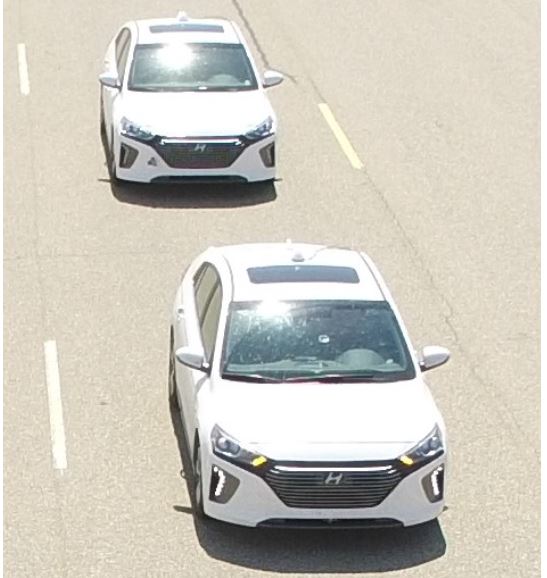}
\caption{Picture of two Hyundai Ioniq PHEVs used as our test platform}
\label{fig:ioniq_pic} 
\end{figure}
Our test platform is a set of two identical 2017 Hyundai Ioniq plug-in hybrid electric vehicles (PHEV) depicted in Figure~\ref{fig:ioniq_pic}.
The weight of the two cars is matched before every experimental session.
The hybrid powertrain has a pre-transmission parallel architecture \cite{Guzzella2013}. Both vehicles are instrumented with a Mando camera, a Delphi ESR radar, a Cohda MK5 On-Board Unit (including a DSRC radio and a GPS unit), a high accuracy fuel flowmeter, and high accuracy current sensors to measure the high voltage battery current, the motor current, the starter/generator current, and the current for auxiliaries.
The longitudinal motion can be controlled manipulating the reference torque of the electric motor, the internal combustion engine, and the hydraulic brakes.
In general, production adaptive cruise control only allows a minimum time gap of about \SI{1}{\s}.
Since our goal is to characterize the effect of short inter-vehicular distance on $C_x$, in our experiments we override the production systems and implement a dedicated longitudinal control, that allows smaller time gap.
Figure~\ref{fig:platoon_scheme} depicts the vehicle setup for our experiment.

\begin{figure}
\centering
\includegraphics[width=\linewidth,height=\textheight,keepaspectratio]{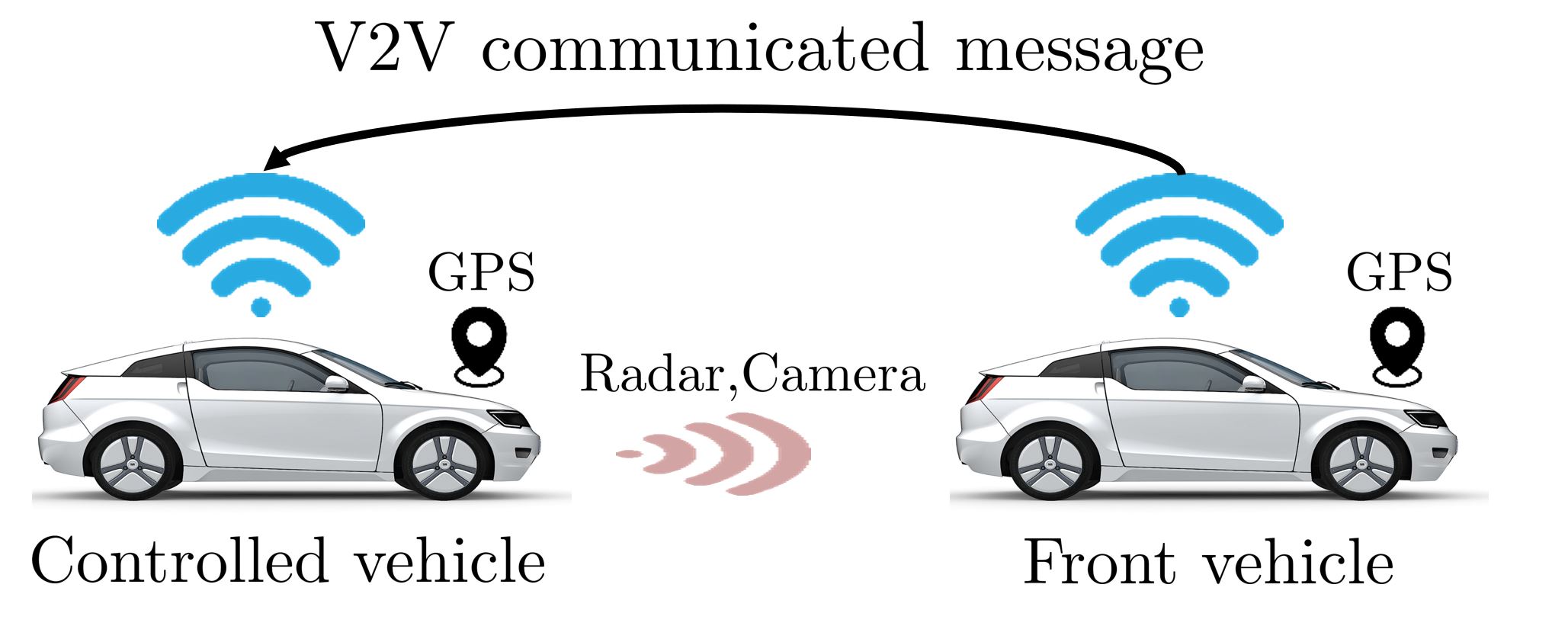}
\caption{Platoon of two vehicles with predecessor following communication topology}
\label{fig:platoon_scheme} 
\end{figure}

Our analysis is focused on how the inter-vehicular distance affects the wheel torque required to drive at the same speed and road grade.
Our test platform does not include a direct measurement of the wheel power, but only an estimate thereof based on motor and engine estimated torques.
On the other hand, accurate measurements of fuel and electric energy consumption are available; hence the energy performance \emph{after} the powertrain can be accurately characterized in terms of total battery and/or fuel consumption, although we are primarily interested in the energy performance \emph{before} the powertrain in terms of total wheel power consumption.
To minimize the effect of powertrain operation on the data collected in our experiments, we override the production powertrain control and operate a (very simple) dedicated powertrain controller.

The dedicated powertrain controller allows only one prime mover (engine or motor) at a time; in other words, the vehicles either operate in purely electric (or full electric, FE) mode or in purely thermal (or full combustion engine, FC) mode.
Figure~\ref{fig:phev_powertrain} depicts the powertrain architecture and these two configurations.
This simple strategy minimizes energy flows between the powertrain components, which cause transient operation and are difficult to measure and model.
Since the scope of this section is the characterization of the longitudinal dynamics and specifically of the aerodynamic drag, we only consider portions of data in which the wheel speed and torque are positive, $T_w > 0, \omega_w > 0$.
We also denote by $r_g$ the gear ratio between the wheels and the powertrain axle (which is the same for the motor and the engine), by $P_a$ the electric power absorbed by auxiliaries such as air conditioning, and by $\mu \in \left\{ \textnormal{FE}, \textnormal{FC}\right\}$ the operating mode of the powertrain.
With the assumptions and limitations just stated, we can model the battery and fuel power as
\begin{equation}
\begin{aligned}
    P_b &= 
    \begin{cases}
        \dfrac{T_w \omega_w}{\eta_m(T_w / r_g , \omega_w r_g)} + P_a , \; &\text{if} \; \mu = FE, \\
        0 , \; &\text{if} \; \mu = FC, \\
    \end{cases} \\
    P_f &= 
    \begin{cases}
        P_a , \; &\text{if} \; \mu = FE, \\
        \dfrac{T_w \omega_w}{\eta_e(T_w / r_g , \omega_w r_g)} , \; &\text{if} \; \mu = FC, \\
    \end{cases}
\end{aligned}
\label{eq:pt_model}
\end{equation}
where $\eta_m$ and $\eta_e$ are the efficiencies of the electric motor and internal combustion engine, respectively; both are assumed to be known static functions of the motor and engine torque and speed.

\begin{figure}
\centering
\begin{subfigure}{\columnwidth}
\centering
\begin{tikzpicture}
\node[draw][minimum width=.6cm,minimum height=.6cm] (rect) at (4.2,9) (W) {W};
\node[draw][minimum width=.6cm,minimum height=.6cm] (rect) at (3.2,9) (T) {T};
\node[draw][minimum width=.6cm,minimum height=.6cm] (rect) at (1.6,9) (M) {M};
\node[draw][minimum width=.6cm,minimum height=.6cm] (rect) at (0,9) (B) {B};
\node[draw,color=gray,dashed][minimum width=.6cm,minimum height=.6cm] (rect) at (2.4,8.2) (E) {E};
\node[draw][minimum width=.6cm,minimum height=.6cm] (rect) at (0.8,8.2) (A) {A};
\node[draw][circle,minimum size=.1cm] at (2.4,9) (S1) {};
\node[draw][circle,minimum size=.1cm] at (0.8,9) (S2) {};
\draw[<->] (T) -- node[midway,above] {} (W); 
\draw[<->] (S1) -- node[midway,above] {} (T); 
\draw[<->] (M) -- node[midway,above] {} (S1); 
\draw[<->] (S2) -- node[midway,above] {} (M); 
\draw[<->] (B) -- node[midway,above] {} (S2); 
\draw[->] (E) -- node[midway,right] {} (S1); 
\draw[->] (S2) -- node[midway,right] {} (A); 
\end{tikzpicture}
\caption{Full electric (FE) configuration}
\end{subfigure}
\\[\baselineskip]
\begin{subfigure}{\columnwidth}
\centering
\begin{tikzpicture}
\node[draw][minimum width=.6cm,minimum height=.6cm] (rect) at (4.2,9) (W) {W};
\node[draw][minimum width=.6cm,minimum height=.6cm] (rect) at (3.2,9) (T) {T};
\node[draw,color=gray,dashed][minimum width=.6cm,minimum height=.6cm] (rect) at (1.6,9) (M) {M};
\node[draw][minimum width=.6cm,minimum height=.6cm] (rect) at (0,9) (B) {B};
\node[draw][minimum width=.6cm,minimum height=.6cm] (rect) at (2.4,8.2) (E) {E};
\node[draw][minimum width=.6cm,minimum height=.6cm] (rect) at (0.8,8.2) (A) {A};
\node[draw][circle,minimum size=.1cm] at (2.4,9) (S1) {};
\node[draw][circle,minimum size=.1cm] at (0.8,9) (S2) {};
\draw[<->] (T) -- node[midway,above] {} (W); 
\draw[<->] (S1) -- node[midway,above] {} (T); 
\draw[<->] (M) -- node[midway,above] {} (S1); 
\draw[<->] (S2) -- node[midway,above] {} (M); 
\draw[<->] (B) -- node[midway,above] {} (S2); 
\draw[->] (E) -- node[midway,right] {} (S1); 
\draw[->] (S2) -- node[midway,right] {} (A); 
\end{tikzpicture}
\caption{Full combustion engine (FC) configuration}
\end{subfigure}
\caption{Architecture of a pre-transmission hybrid electric powertrain and its two configurations used for our data collection campaign. W: longitudinal dynamics. T: transmission. E: internal combustion engine. M: electric motor. B: high-voltage battery. M is permanently connected to T. E is connected to M and T through a clutch (not shown here); when the clutch is closed, M and E and the input shaft of T are on the same axle; A: auxiliary load.}
\label{fig:phev_powertrain}
\end{figure}
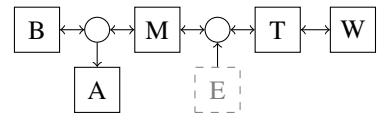
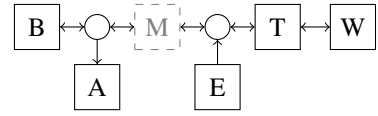

Each experiment is performed at a fixed combination of speed, time gap, and powertrain mode.
Table~\ref{tab:experment-summary} summarizes the amount of data collected (in terms of number of laps around a \SI{10}{\km} long oval test track) for each combination.
Only data from complete laps is utilized.

Laps in full electric mode are constrained by the limited on-board storage of electric energy (\SI{8.9}{\kWh}); in most operating conditions, 3 complete laps can be performed in FE mode before a battery recharge is necessary.
Note that, in order to explore energy saving even at a very small inter-vehicular distance, we conducted experiment with the \SI{0.3}{s} time gap although this small gap is not generally accepted for safety reasons in commercial, human driven vehicles.

\begin{table}[]
\caption{Number of test laps for each powertrain mode (FE/FC), time gap (columns), and velocity (rows).}
\centering
\begin{tabular}{cccccc}
\toprule
& $\infty$  & \SI{2}{s} & \SI{1}{s} & \SI{0.5}{s}   & \SI{0.3}{s}   \\
\midrule
\SI{24.6}{\meter/\second}    & FE: 1     & FE: 0     & FE: 1     & FE: 1         & FE: 1        \\
(\SI{55}mph)& FC: 2     & FC: 2     & FC: 2     & FC: 2         & FC: 2        \\
\midrule
\SI{29.1}{\meter/\second}    & FE: 1     & FE: 0     & FE: 1     & FE: 1         & FE: 1        \\
(\SI{65}mph)& FC: 2     & FC: 2     & FC: 2     & FC: 2         & FC: 2        \\
\midrule
\SI{33.5}{\meter/\second}    & FE: 1     & FE: 0     & FE: 1     & FE: 1         & FE: 1        \\
(\SI{75}mph)& FC: 2     & FC: 2     & FC: 2     & FC: 2         & FC: 2        \\
\bottomrule
\end{tabular}
\label{tab:experment-summary}
\end{table}

The test facility used for these experiments is the Hyundai-Kia California Proving Ground.
The tests discussed here are performed on a single lane in a \SI{10}{\km} long, oval test track.
Note that because we limit the scope of this project to the longitudinal motion of CAVs excluding their lateral motion, our tests were performed on a single lane in the oval test track, of which the curvature is always less than \SI{500}{\per\m}. The lateral motion of CAVs was controlled by the human drivers who tried to maintain the vehicles at the center of the lane throughout the experiments.
The facility includes an apparatus to measure the wind speed and direction; all the data presented next are collected in conditions of low wind speed (average and maximum wind speed always less than \SI{3.18}{\m\per\s} and \SI{3.95}{\m\per\s}, respectively, which is below the limits set in the SAE J1263 standard \cite{J1263_201003}).
All experiments are performed in a time window of 4 days.
Moreover, throughout every experiment,  the auxiliary load on both vehicles was maintained at the same level.

\subsection{Energy consumption results}
\begin{figure*}
\centering
\includegraphics[width=\textwidth]{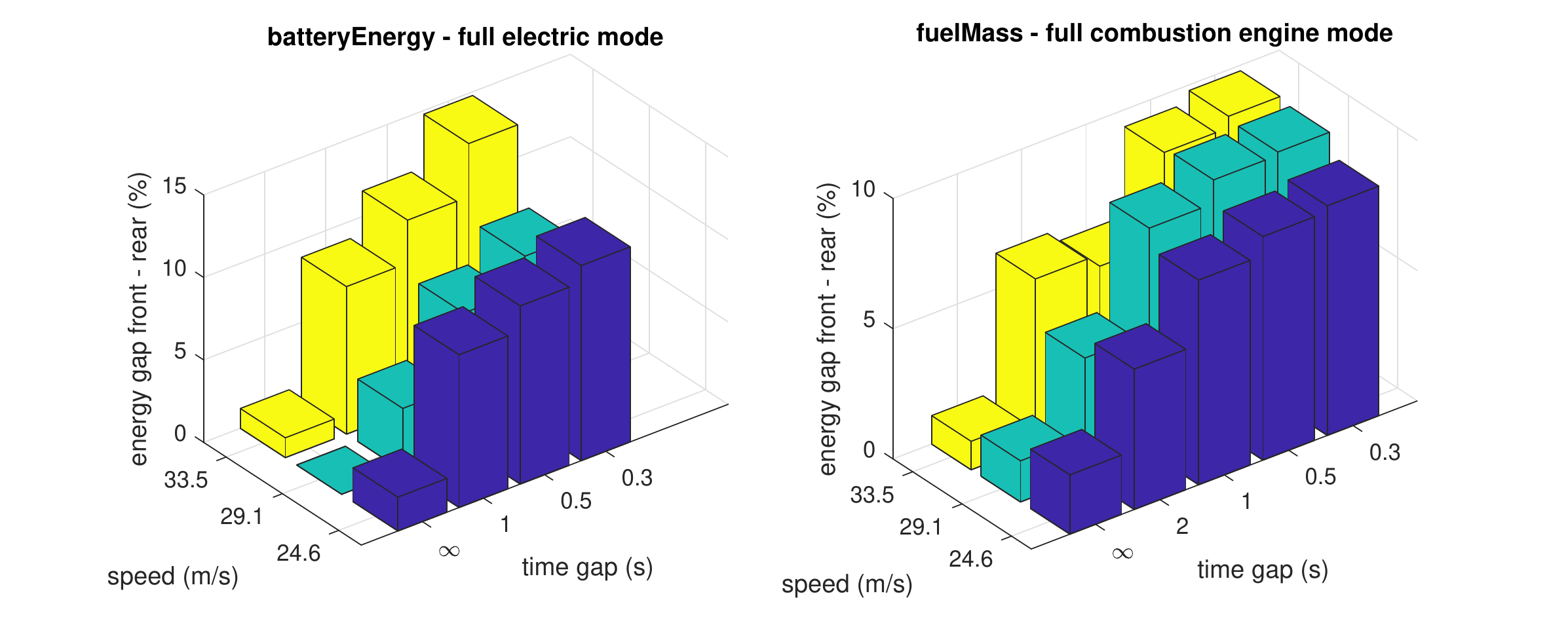}
\caption{Percentage variation of electric energy consumed (rear vs front) in full electric mode (left, $\mu=FE$) and fuel consumed (rear vs front) in full combustion engine mode (right, $\mu=FC$), as a function of speed and time gap.}
\label{fig:drag:cdcs}
\end{figure*}

Figure~\ref{fig:drag:cdcs} shows the energy consumption results from the tests listed in Table~\ref{tab:experment-summary}.
Note that these results are obtained with a homogeneous two-vehicle platoon composed of Hyundai Ioniq PHEVs.
Clearly, results will be different for other vehicles or the heterogeneous platoon.
The left chart refers to data collected in full electric mode ($\mu = FE$), while the right chart refers to data collected in full combustion engine mode ($\mu = FC$).
Each chart shows time gap (in seconds) and vehicle speed (in meters per second) on the horizontal axes, and the percentage variation of energy consumption between front and rear vehicle on the vertical axis.
The percentage variation is computed on the integral of battery power $P_b$ for full electric data, and on the integral of fuel power $P_f$ for full combustion engine data.

In Figure~\ref{fig:drag:cdcs}, the time gap is labeled as $\infty$ for experiments performed on isolated vehicles; these tests were performed to gather baseline data, as well as for comparing the performance of the two vehicles.
Both figures show that, in isolated conditions and at different speeds, the fuel/energy consumption of the two test vehicles differ by at most \SI{3}{\percent}.
Such difference can be attributed to small variability between the two  vehicles.

The collected data enable direct evaluation of the energy performance variation (improvement) after the powertrain.
In full combustion engine mode, the fuel consumption is reduced between \SI{4}{\percent} and \SI{10}{\percent}.
In full electric mode, the electric energy consumption is reduced between \SI{4}{\percent} and \SI{15}{\percent}.
In both powertrain modes, the air drag appears to consistently reduce with time gap, for a fixed speed except one instance.
At the speed of \SI{33.5}{\m \per \second}, the fuel energy consumption was slightly reduced (less than $0.5\%$) when the time gap increased from \SI{1}{s} to \SI{2}{s}.
We believe that this is due to imperfect lateral alignment between two vehicles. All the other parameters such as wind, weight, auxiliary load are either corrected for or kept within small bounds.

\subsection{Interpretation of the results}

The data presented above are purely based on our measurements, in the listed operating conditions.
We now provide an interpretation of the same data in terms of the vehicle longitudinal model \eqref{eq:long_dyn} and powertrain model \eqref{eq:pt_model}.
In essence, we propose a simple model which generalizes the results above to any operating condition.
To this aim, we use an existing model of the powertrain components (to connect the measurements of the energy consumption to the estimated wheel torque) and we fit a longitudinal vehicle model.

First, we show that our estimated wheel torque and our powertrain model explain well the energy consumption values that we found in the experiments.
For all the data points in our measurements, we compute the predicted consumption of battery energy $\hat{P}_b$ and fuel energy $\hat{P}_f$ according to model \eqref{eq:pt_model}, using the measured wheel torque $T_w$, wheel speed $\omega_w$, auxiliary power $P_a$, and powertrain mode $\mu$ as inputs.
Figure~\ref{fig:pt_histograms} shows the histograms of the normalized residuals
\begin{align*}
e_b = \frac{P_b - \hat{P}_b}{RMS(P_b - \hat{P}_b)} , \\
e_f = \frac{P_f - \hat{P}_f}{RMS(P_f - \hat{P}_f)} .
\end{align*}

\begin{figure}
\centering
\includegraphics[width=\columnwidth]{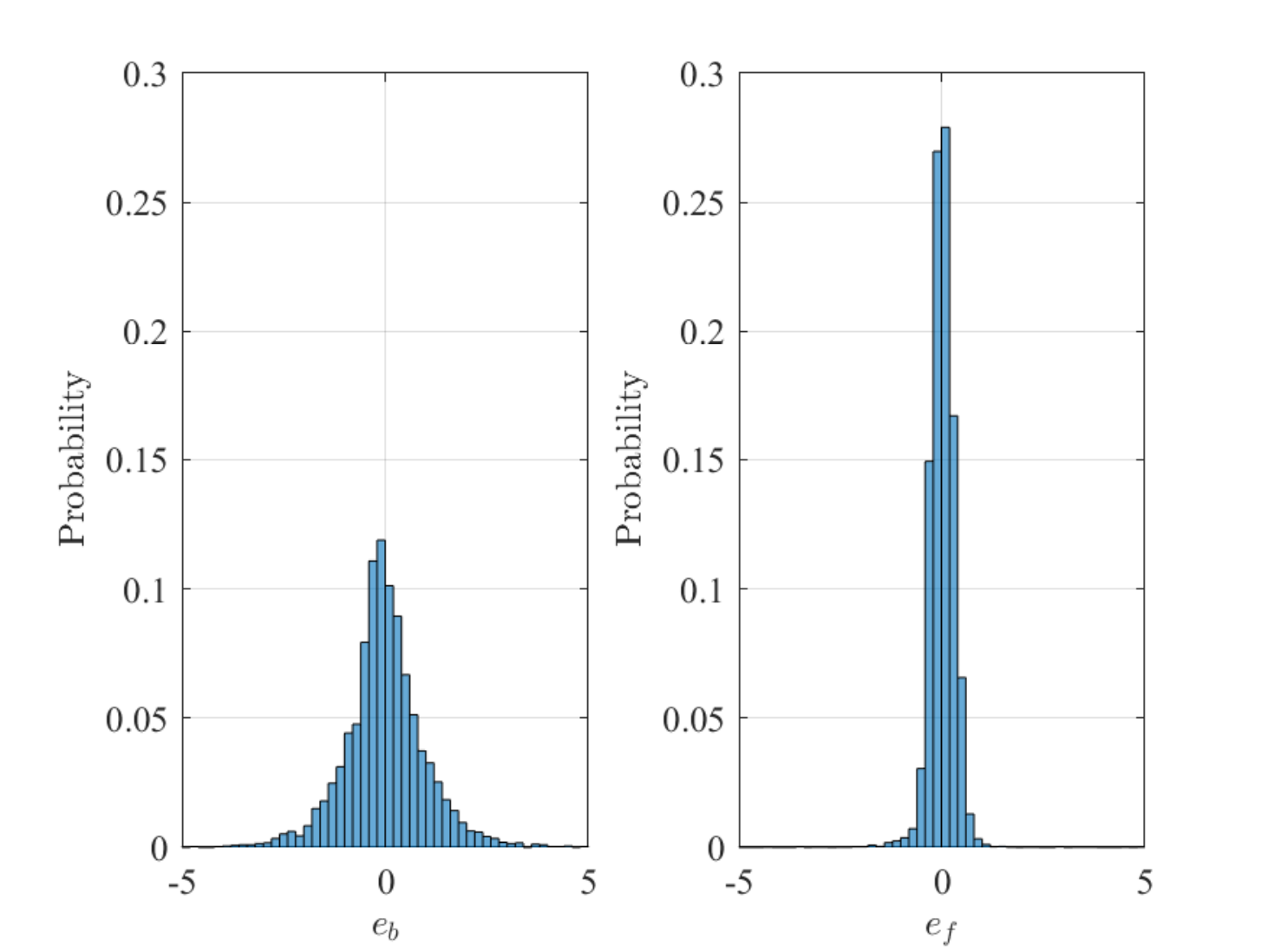}
\caption{Histograms of the normalized residuals $e_b$ and $e_f$.}
\label{fig:pt_histograms}
\end{figure}

Now, we fit the longitudinal model \eqref{eq:long_dyn}.
The model parameters $C_r, C_v, C_x$ are known to depend on a number of factors, including gear ratio, vehicle speed, tire pressure, road surface, and inter-vehicle distance \cite{Guzzella2013}.
We assume that:
\begin{itemize}
    \item the rolling coefficient $C_r$ is constant;
    \item the viscous coefficient $C_v$ is constant;
    \item the air drag coefficient is a function of inter-vehicle distance, $C_x = C_x(d)$.
\end{itemize}

To fit the model with mode- and distance-dependent parameters, we first define the cost function
\begin{equation*}
\begin{split}
    J(\bar{d},\bar{\mu}) &= \\
    &\lVert F_w(k) -  m a(k) - m g C_r - C_v v(k) \\
    &- \frac{1}{2} \rho A C_x(\bar{d}) v(k)^2 \rVert , \\ &\forall k \in \mathcal{K}(\bar{d},\bar{\mu}) := \left\{ k | d(k) = \bar{d}, \mu(k) = \bar{\mu} \right\} ,
\end{split}
\end{equation*}
where $k$ is a datapoint index, $\bar{d}$ and $\bar{\mu}$ are the values of (target) $d$ and $\mu$ allowed by Table~\ref{tab:experment-summary} (in symbols $\bar{d} \in \mathcal{D}$ and $\bar{\mu} \in \mathcal{M}$), and $\mathcal{K}$ is the cluster of all datapoint indexes corresponding to $\bar{d}$ and $\bar{\mu}$.
We formulate the following optimization problem:
\begin{equation}
\label{eq:vd-cx-fitting}
\begin{aligned}
    & &\underset{C_r, C_v, C_{x}(\bar{d}), \forall \bar{d} \in \mathcal{D}, \forall \bar{\mu} \in \mathcal{M}}{\min} &\sum_{\bar{d} \in \mathcal{D}, \bar{\mu} \in \mathcal{M}} J(\bar{d},\bar{\mu}) \\ 
    &  &\text{s.t.}   &C_r \geq 0 , \\
    &  &   &C_v \geq 0 , \\
     &  &   & C_{x}(\bar{d}) \geq 0 , \forall \bar{d} \in \mathcal{D}\\
\end{aligned}
\end{equation}
where $C_x(\bar{d})$ is defined using the vehicular air drag reduction model from \cite{hucho2013aerodynamics} as 
\begin{equation}
\label{eq:cx-model}
    C_x(d) = C_{x,0} \bigg( 1 - \dfrac{C_{x,1}}{d + C_{x,2}} \bigg) .
\end{equation}
We parse the problem with YALMIP \cite{Lofberg2004}, solve it with IPOPT \cite{wachter2009ipopt} on \SI{80}{\percent} of the data, and find the parameters summarized in Table~\ref{tab:fitted-param}.

\begin{table}[]
\caption{Known (top) and fitted (bottom) parameters for the rear vehicle.}
\centering
\begin{tabular}{ll}
\toprule
Model parameter & Value \\
\midrule
$M$ & \SI{1844}{\kg} \\
$R_w$ & \SI{0.288}{\m} \\
$\rho$ & \SI{1.206}{\kg\per\cubic\m} \\
$A$ & \SI{2.629}{\square\m} \\
\midrule
$C_r$ & 0.0093 \\
$C_v$ & 0 \\
$C_x(d)$ & $ 0.3350 \bigg( 1-\dfrac{68.3193}{d+142.4522} \bigg)$ \\
\bottomrule
\end{tabular}
\label{tab:fitted-param}
\end{table}


Figure~\ref{fig:vd_histogram} shows the histogram of the normalized residuals for the validation dataset:
\begin{equation*}
e_w = \frac{T_w - \hat{T}_w}{RMS(T_w - \hat{T}_w)} .
\end{equation*}
For all the validation datapoints, we compute the predicted wheel torque $\hat{T}_w$ according to model \eqref{eq:long_dyn} and the parameters in Table~\ref{tab:fitted-param}, using the measured speed $v$ and inter-vehicle distance $d$ as inputs ($\vartheta = 0$ throughout our experiments).


\begin{figure}
\centering
\includegraphics[width=\columnwidth]{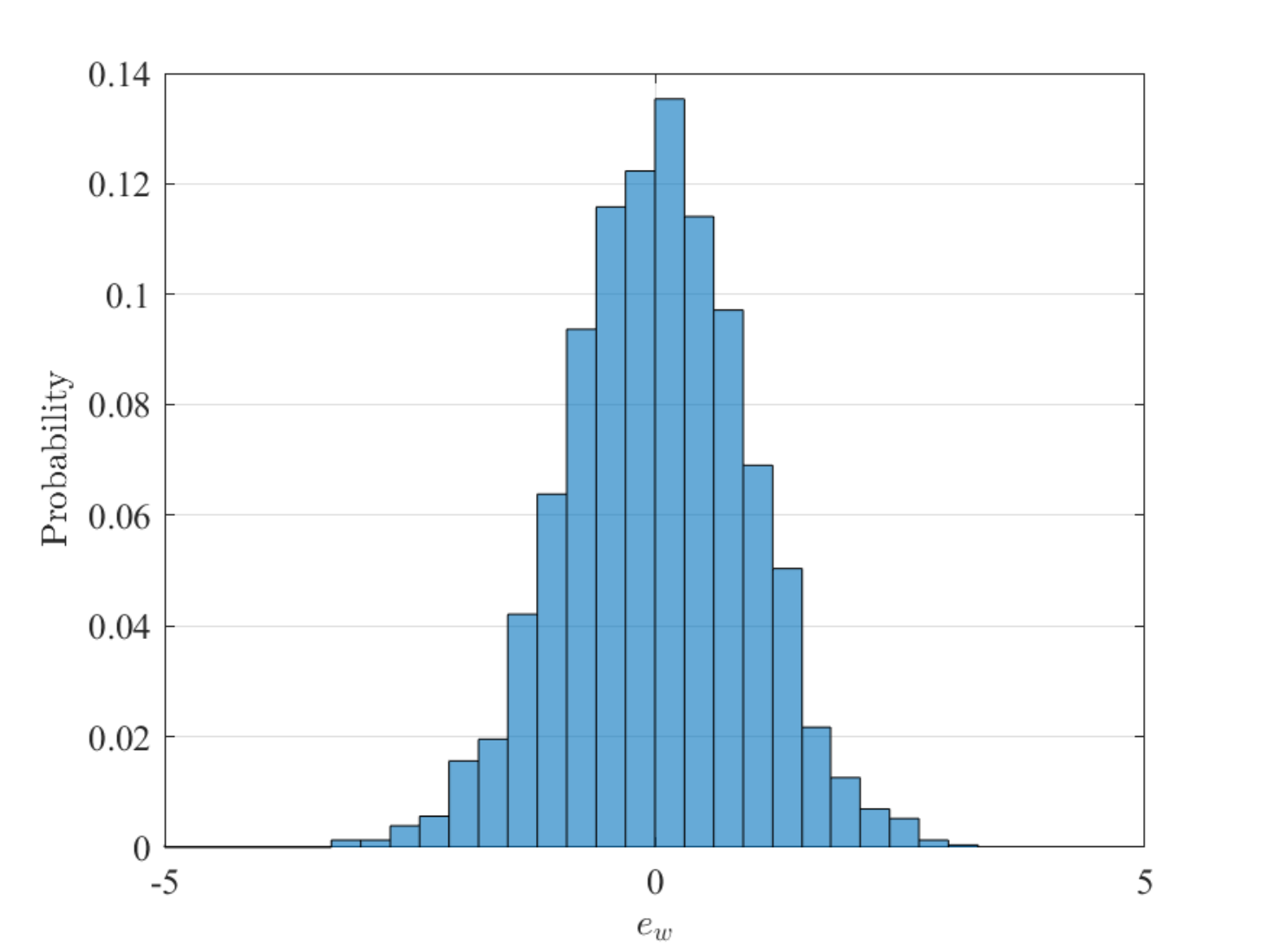}
\caption{Histogram of the normalized residuals $e_w$.}
\label{fig:vd_histogram}
\end{figure}


\section{Cooperative adaptive cruise control with model predictive control}
\label{sec:cacc_design}

We have shown in Section~\ref{sec:veh_dyn_model} that reducing the inter-vehicular gap can reduce, at least at steady state, the energy consumption of the tail vehicle(s).
In this section, we propose a CACC design which exploits this finding to save energy, i.e. a CACC that enhances energy efficiency (and possibly road throughput) by operating vehicles at small relative distance, while avoiding collisions. 
In particular, we investigate how leveraging V2V communication, which can offer a fast yet accurate measurement of the front vehicle position, velocity, acceleration, and potentially forecasts thereof, leads to the energy saving which is quantified by using the experimental performance in Figure~\ref{fig:pt_histograms}. 
It is noted that this section only deals with the simulation results. Readers can refer to \cite{Youtube_plat_coord, smith2020improving} for the visual demonstration of our compact platoon system.

Our CACC controller oversees the interaction between a vehicle (referred to as the \textit{controlled vehicle}) and the vehicle ahead (referred to as the \textit{front vehicle}) so that their inter-vehicular gap is minimized.
The front vehicle sends to the controlled vehicle a V2V message which can be useful in shortening this gap. 
Our proposed control design is based on Robust MPC which can systematically handle multiple constraints in an uncertain system.




\subsection{Control Oriented Modeling}
\label{sec:control_oriented_msg}

In model \eqref{eq:long_dyn}, the aerodynamic drag coefficient $C_x$ is coupled to the distance to the front vehicle, as seen in \eqref{eq:cx-model}, while the road grade $\theta$ depends on the vehicle absolute position.
Using such a model in an MPC framework requires a strategy to solve the resulting nonlinear optimization problem in real time and considering that the solver may converge to local minima. 
We introduce a control oriented model, in which $C_x$ is modeled as \eqref{eq:cx-model} and $\theta$ is approximated as a constant along the MPC prediction horizon. 

\subsubsection{Control oriented vehicle dynamics} 
\label{sec:cont_oriented_veh_dyn}

We consider the problem of controlling a single vehicle and its interaction with the front vehicle. Therefore, the state variables are the distance to the front vehicle, the velocity of the controlled vehicle, and the velocity of the front vehicle, denoted as $d$, $v$, and $v^f$, respectively.
The input is the wheel torque of the controlled vehicle, denoted as $T^w$, and the front vehicle acceleration is an uncertain variable, denoted as $a^f$.
The control oriented model of the vehicle longitudinal dynamics is a switched uncertain system; at time $t$, the one-step ahead prediction is expressed as 
\begin{subequations}
\label{eq:cont_oriented_long_dyn}
\begin{align}
d(k+1 |t) &= d(k |t) + t_s(v^f(k |t) - v(k |t)),\\
v(k+1 |t) &= \bigg( v(k |t) + \frac{t_s}{m}  \Big(\frac{T^w(k |t)}{R_w}  \nonumber \\
& \qquad - mg(\sin{\vartheta(t)} + {C}_r\cos{\vartheta(t)}) \nonumber \\
& \qquad- \frac{1}{2} \rho A C_x(d(k |t))  v(k |t)^2 \Big) \bigg)^+  \label{eq:ego_veh_dyn},\\
v^f(k+1 |t) &=  \Big( v^f(k |t) + t_s a^f(k |t) \Big)^+, \label{eq:front_veh_dyn}
\end{align}
\end{subequations}
where $\vartheta(t)$ is a constant approximated at time $t$, and $C_x(d(k |t))$ is from \eqref{eq:cx-model}. 
The $(\bullet)^+$ operators in the controlled vehicle velocity equation \eqref{eq:ego_veh_dyn} and the front vehicle velocity equation \eqref{eq:front_veh_dyn} are to prevent vehicle velocities from taking negative values.
In the remainder of the paper, the dynamic model \eqref{eq:cont_oriented_long_dyn} is denoted by
\begin{equation}
x(k+1 |t) = f_t(x(k|t), u(k|t), a^f(k|t)) ,
\label{eq:vehicle_model_short}
\end{equation}
where $x(k|t) = [d(k|t),\, v(k|t), \, v^f(k|t)]^{\top}$ and $u(k|t) = T^w(k|t)$.

\subsubsection{State and input constraints and uncertainty bounds} 
\label{sec:constraints}

State constraints are enforced to prevent collision with the front vehicle and violation of speed limits. Input constraints are enforced to account for the physical limitations of the actuators of the controlled vehicle.
The state and input constraints are compactly expressed in the form
\begin{subequations}\label{eq:stateinputconstraints}
\begin{align}
x(k|t) \in \mathbb{X}&:= \{(d(k|t),v(k|t)):  d_{\textnormal{min}} \leq d(k|t),\,  \nonumber \\
& \qquad \qquad \qquad \qquad \qquad \,\, 0 \leq v(k|t) \leq v_{\textnormal{max}}  \}, \, \label{eq:stateconstraints}\\
u(k|t) \in \mathbb{U}&:= \{T^w(k|t): T^w_{\textnormal{min}} \leq T^w(k|t) \leq T^w_{\textnormal{max}} \}. \label{eq:inputconstraints}
\end{align}
\end{subequations}
where $d_{\textnormal{min}}$ denotes the minimum safety distance and  $v_{\textnormal{max}}$ is the maximum velocity.

The front vehicle acceleration is an uncertain, bounded variable:
\begin{equation}
\label{eq:a_front_set}
a^f(k|t) \in \mathbb{A}^f :=\{ a^f(k|t) : a^f_{\textnormal{min}} \leq a^f(k|t) \leq
a^f_{\textnormal{max}} \}.
\end{equation}
As we further discuss in the remainder of the paper, the uncertainty bounds $a^f_{\textnormal{min}}, a^f_{\textnormal{max}}$ can be fixed and assumed \emph{a priori} by the controlled vehicle, or can be communicated by the front vehicle and adapted to the current operating condition, if an appropriate V2V communication protocol has been established.
In the latter case, the uncertainty set $\mathbb{A}^f$ can expand or shrink, based on the maximum braking and acceleration bounds transmitted by the front vehicle via V2V communication.

\subsection{V2V message structure}
\label{sec:v2v_msg}
We now detail a possible V2V message structure, which can carry the information required to positively affect the vehicle energy consumption.
Under the predecessor-following communication topology \cite[Fig. 4]{guanetti2018control}, V2V messages are only sent from a vehicle to the vehicle immediately following it.
In our setup, the controlled vehicle receives, from the front vehicle, a V2V message $m^f$, composed as
\begin{equation} \label{eq:v2v_msg}
\begin{aligned}
m^f(t) = \{ s^f&(t-h), \,v^f(t-h),\,\mathbb{A}^f, \nonumber \\
\,& \,[\hat{a}^f(t-h|t-h),\dots,\hat{a}^f(t-h+N_T|t-h)]  \}
\end{aligned}
\end{equation}
where $h$ is a communication delay; $s^f(t-h)$ and $v^f(t-h)$ are the absolute location and the velocity of the front vehicle at time $t-h$; 
$\mathbb{A}^f$ is the acceleration bounds as expressed in \eqref{eq:a_front_set}.
$N_T$ is defined as the trust horizon, which indicates the number of time steps for which the front vehicle trusts its acceleration forecast to take the same values of its true acceleration, i.e. $\hat{a}^f(j|t-h) = \bar{a}^f(j) , \, j = t-h , ... , t-h + N_T$.

\subsection{Model Predictive Control Formulation}
\label{sec:mpc_form}

We formulate the CACC problem as the following constrained finite horizon optimal control problem at time $t$.
\begin{subequations}\label{eq:mpcproblem}
\begin{align}
&\displaystyle \min_{u_{\cdot |t}} \,\,  
\sum_{k=t+1}^{t+N_p} \norm{ {d}(k |t) - d_\textnormal{min}}_{Q}^2 + \sum_{k=t}^{t+N_p-1} \norm{u(k |t)}_{R}^2 \nonumber\\
&\;\;\, \qquad +\sum_{k=t+1}^{t+N_p-1} \norm{u(k |t) - u(k-1|t)}_{D}^2 \label{eq:cost_fcn}\\
& \textnormal{subject to} \nonumber \\
&\;\;\, {x}(t|t) = \Tilde{x}(t|t), \label{eq:init_condition}\\
&\;\;\, {x}(k+1|t) = f_{t}( {x}(k|t), u(k|t), \Tilde{a}^f(k)) \nonumber \\ &\;\;\, \qquad\qquad\qquad\qquad\qquad\qquad \forall k \in [t,...,t+N_p-1], \label{eq:pred_system_update} \\ 
&\;\;\, {x}(k|t) \in \mathbb{X}, \, u(k|t) \in \mathbb{U} \quad\forall k \in [t,...,t+N_p-1], \label{eq:system_cons}\\
&\;\;\, {x}(t+N_p|t) \in \mathbb{C} , \label{eq:term_cons}
\end{align}
\end{subequations}%
where the cost function \eqref{eq:cost_fcn} penalizes deviations from the minimum safety distance $d_{\textnormal{min}}$, the actual excitation, and the jerk, with the weighting terms $Q$, $R$, and $D$, respectively.
Note that because the goal of the controller is to minimize the distance gap with the front vehicle in order to save energy, we use the minimum distance $d_{\textnormal{min}}$ to be the desired tracking distance. 
For the same purpose, $Q$ is set to be much larger than $R$ and $D$.

Moreover, it is noted that in order to guarantee string stability for the system controlled by our time-domain MPC control design, additional constraints can be imposed. These include restricting the magnitude of the controlled vehicle's acceleration within that of the preceding vehicle \cite{kianfar2011receding} or imposing the move-suppression distance constraint \cite{dunbar2012distributed}.

The initial state is set in \eqref{eq:init_condition} where $\Tilde{x}(t|t)$ is the estimated state at time $t$ using the available information at time $t$.
In a vehicle without connectivity, distance to and velocity of the front vehicle are only available via on-board sensors such as radar or camera.
However, these measurements are subject to significant delay and noise \cite{floudas2005survey}.
V2V communication can reduce the noise and the delay in the perception of the front vehicle down to the communication latency.
In order to compensate for any perception delay (either from measurement or V2V communication latency), we use the following shifting equation:
\begin{subequations}
\label{eq:rob_state_pred}
\begin{align}
&\Tilde{x}(t-h|t) = \begin{bmatrix} \hat{d}(t)  \\ \hat{v}(t)^+ \\ (\hat{v}^f(t) )^+ \end{bmatrix} \label{eq:rob_ob_init} \\
&\Tilde{x}(k+1 |t) = \begin{bmatrix} \Tilde{d}(k |t) + t_s(\Tilde{v}^f(k |t) - \Tilde{v}(k |t))   \\ \hat{v}(k+1) \\ \Big( \Tilde{v}^f(k |t) + t_s \Tilde{a}^f(k |t) \Big)^+ \end{bmatrix} \nonumber \\
& \qquad \forall k=t-h, t-h+1, \dots, t-1 , \label{eq:rob_ob_dyn}
\end{align}
\end{subequations}
where $\hat{d}(t):=\bar{d}(t-h)+n^d(t-h)$ and $\hat{v}^f(t):=\bar{v}^f(t-h)+n^{v^f}(t-h)$ are the front vehicle distance and velocity with a known, constant, non-negative delay $h$ either measured or communicated at time $t$, respectively;
$n^d(t-h)$ and $n^{v^f}(t-h)$, such that $|n^d(t-h)| \leq n^d_{\textnormal{max}}$ and $|n^{v^f}(t-h)| \leq n^{v^f}_{\textnormal{max}}$, are the noises for distance and front vehicle velocity measurements, respectively;
$\hat{v}(t):=\bar{v}(t)$ is the controlled vehicle velocity accurately and immediately measured by an on-board speed sensor obtained at time $t$; 
$\bar{u}(k)$ is the input at time step $k$;
$\Tilde{a}^f(\cdot)$ is defined as
\begin{equation} 
\label{eq:front_a_def}
\Tilde{a}^f(k) = 
\begin{cases}
\hat{a}^f(k|t-h) , &\mbox{if} \, k \leq t-h+N_T , \\
a^f_{\textnormal{min}} , & \mbox{else} ,
\end{cases}
\end{equation}
where $\hat{a}^f(k|t-h)$ is the communicated front vehicle acceleration forecast at time $k$ obtained at time $t$; $a^f_{\textnormal{min}}$ is the lower bound for the set $\mathbb{A}^f$; $N_T$ is the trust horizon.
In the case of ACC without V2V communication, the front vehicle acceleration forecast and its bounds are not available. Hence, the front vehicle acceleration forecast  $\Tilde{a}^f(\cdot)$ is assumed to be the conservative (over-estimated) value of the maximum braking at all times.
Note that, with the simple structure of the system dynamics \eqref{eq:cont_oriented_long_dyn}, obeying the constraints \eqref{eq:stateinputconstraints} with the front vehicle acceleration \eqref{eq:front_a_def} is a sufficient condition to guarantee \eqref{eq:stateinputconstraints} for any front vehicle acceleration in $\mathbb{A}^f$.  
Also, in case of the time-varying delay, the author suggests to use the upper bound of the time-varying delay in place of $h$. This will add the conservatism to the system but it can still have the capability to guarantee to maintain above the minimum safety distance although it can fail to ensure that the vehicle does not exceed the maximum vehicle speed as the system would be under-estimating the air drag.

\eqref{eq:pred_system_update} represents the vehicle dynamics update defined by \eqref{eq:vehicle_model_short};
\eqref{eq:system_cons} represents the system state and input constraints \eqref{eq:stateinputconstraints};
\eqref{eq:term_cons} is the robust control invariant terminal set constraint detailed in Appendix. While the problem \eqref{eq:mpcproblem} is formulated as a mixed integer optimization (MIP) problem due to the fact that the set $\mathbb{C}$ is the union of polyhedrons at each different front vehicle velocity, it can actually be casted as a non-integer program to lessen the computational burden of MIP. 
The front vehicle velocity of the front vehicle is offered a-priori  (by the combination of the V2V forecast and the robust velocity predictor) to the computation or solving of the problem. 
This allows the problem to know the minimum possible front vehicle velocity at the end of the horizon and which polyhedron to use in the terminal set. This way, the computational burden is light enough for the online controller implementation. Our previous experimental test proves the feasibility of the real-time implementation of the similar control designs using the DSPACE ECU \cite{turri2017model, smith2020improving, bae2019real}.

Note that due to the simple structure of the system \eqref{eq:cont_oriented_long_dyn}-\eqref{eq:a_front_set}, the constraints \eqref{eq:stateinputconstraints} are satisfied for every possible realization of the front vehicle acceleration if the system satisfies \eqref{eq:stateinputconstraints} with the front vehicle acceleration prediction \eqref{eq:front_a_def} \cite{lefevre2016learning}. 

Solving \eqref{eq:mpcproblem}, we have the following optimal inputs and states:
\begin{subequations} \label{eq:optimal_sol}
\begin{align}
u^{*}(t) & = [u^{*}(t|t), u^{*}(t+1|t), ..., u^{*}(t+N_p-1|t)].\label{eq:optimal_sol_inputs} \\
{x}^{*}(t) & = [{x}^{*}(t|t), {x}^{*}(t+1|t), ...,  {x}^{*}(t+N_p|t)], \label{eq:optimal_sol_states}
\end{align}
\end{subequations}
The first input $u^{*}(t|t)$ is applied to the system during the time interval $[t,t+1)$.
At the next time step $t+1$, a new optimal control  problem in the form of \eqref{eq:mpcproblem}, based on new measurements of the state, is solved over a shifted horizon, yielding a \textit{moving} or \textit{receding} horizon control strategy 
\begin{equation}
u(t) = u^{*}(t|t).\label{eq:optimal_input} 
\end{equation}
and the closed loop system is written as 
\begin{equation} \label{eq:cls_sys} 
\bar{x}(t+1) = f(\bar{x}(t), u^{*}(t|t), \bar{v}^f(t)) . 
\end{equation}




\section{Performance analysis of CACC} 
\label{sec:simulation}

In this section we show - through simulations - how the information and predictions shared via V2V communication can be leveraged to reduce the energy consumption in the small platoon of CAVs.
We focus on the value of V2V communication for reducing the inter-vehicular gap, which is linked to the aerodynamic drag, as discussed in Section \ref{sec:veh_dyn_model}.
We consider a homogeneous vehicle platoon of two Hyundai Ioniq PHEVs, the model of which is identified in Section~\ref{sec:veh_dyn_model}, traveling on a level road ($\vartheta = 0$), as depicted in Figure~\ref{fig:platoon_scheme}.
The front vehicle is running at a constant velocity of $\SI{25}{\meter/\second}$; the following vehicle (the controlled vehicle) has an initial velocity of $\SI{15}{\meter/\second}$ and its motion is regulated by the controller \eqref{eq:mpcproblem}-\eqref{eq:optimal_input} where the model $f_t(\cdot)$ is described by the vehicle dynamics \eqref{eq:cont_oriented_long_dyn}; their initial inter-vehicular distance is $\SI{50}{\meter}$.
The two vehicles communicate different message contents, depending on the scenario.
Unless otherwise specified, environment and controller parameters are listed in Table~\ref{tab:model_cont_param}. Note that as a default setting, we assume that the front vehicle can generate the minimum deceleration of $-6\si{m \per \square \s}$ while the only current measurements of its velocity and distance are available without any errors. 
Moreover, Table~\ref{tab:delay-effect}-\ref{tab:Nt-effect} lists the energy consumption compared to that of the front vehicle under different assumptions about on-board sensors and V2V communication during the steady state phase; i.e., the time period after the controlled vehicle reaches a constant distance gap with the front vehicle; the selected metrics are savings on wheel energy, fuel consumption in FC mode, and battery energy in FE mode.
Note that each column represents the percentage of energy spent compared to that of the front vehicle in the unit of the indicated energy source solely.
The fuel and battery energy consumption calculations are interpolated from the experimental data presented in Figure~\ref{fig:drag:cdcs}.


\begin{table*}[!t]
\caption{Environment and controller parameters for different scenarios. VAR: a variable, changing value in each scenario (subsection).}
\centering
\label{tab:model_cont_param}
\begin{tabularx}{0.75\textwidth}{c  c  c  c  c c } 
\toprule
\textbf{Parameter} & \textbf{Description}  & \textbf{Sec~\ref{sec:red_measure_delay}} & \textbf{Sec~\ref{sec:small_dec}} & \textbf{Sec~\ref{sec:trust_vel_forecast}} & \textbf{Sec~\ref{sec:small_vel_uncer}}\\
\midrule
$N_p$ & MPC horizon  & 20 & 20 & 20 & 20  \\
$t_s (s) $ & sampling time   & 0.2 & 0.2 &0.2 & 0.2  \\
$d_{\textnormal{min}} (\si{m}) $ & minimum distance & 5 & 5 & 5 & 5  \\
$v_{\textnormal{max}} (\si{m \per \s} ) $ & maximum velocity & 40 & 40 & 40 & 40  \\
$T^w_{\textnormal{max}} ( \si{Nm}) $ & maximum wheel torque  & 1083 & 1083 & 1083 & 1083  \\
$T^w_{\textnormal{min}} ( \si{Nm}) $ & minimum wheel torque  & -2500 & -2500 & -2500 & -2500  \\
$h$ & measurement delay & VAR & 0 & 0 & 0 \\
$a^f_{\textnormal{min}} (\si{m \per \square \s})$ & front vehicle minimum acceleration & -6 & VAR & -6 & -6  \\
$N_T$ & trust horizon & 0 & 0 & VAR & 0 \\
$n_{\textnormal{max}} (\si{m}, \si{m \per \s})$ & front vehicle measurement noise bound   & 0 & 0 & 0 & VAR  \\
\bottomrule
\end{tabularx}
\end{table*}

\begin{table}
\centering
\caption{Comparison of CACC Controllers. Units for energy consumption are $\%$.}
\begin{subtable}{\columnwidth}
\centering
\begin{tabular}{ c | c| c| c}
\addlinespace
  & Wheel & Fuel in FC mode & Battery in FE mode  \\
\hline 
\hline       front vehicle &  100 & 100.0 & 100.0 \\
$h=2$ &  92.5 & 92.7 & 91.7\\
$h=1$ &  91.6 & 92.3 & 91.3\\
$h=0$ &  90.6 & 91.9 & 90.8 
\end{tabular}
\caption{Effects of measurement delay.}
\label{tab:delay-effect}
\end{subtable}
\vfill
\begin{subtable}{\columnwidth}
\centering
\begin{tabular}{ c | c| c| c}
\addlinespace
& Wheel & Fuel in FC mode & Battery in FE mode \\
\hline 
\hline 
front vehicle &  100.0 & 100.0 & 100.0 \\
$a^f_{\textnormal{min}}=-9\si{m \per \square \s}$ & 93.2 & 93.4 & 92.3\\
$a^f_{\textnormal{min}}=-6\si{m \per \square \s}$ & 90.6 & 91.9 & 90.8\\
$a^f_{\textnormal{min}}=-3\si{m \per \square \s}$ & 87.6 & 91.2 & 89.1
\end{tabular}
\caption{Effects of front vehicle maximum deceleration.}
\label{tab:amin-effect}
\end{subtable}
\hfill
\begin{subtable}{\columnwidth}
\centering
\begin{tabular}{ c | c| c| c}
\addlinespace
& Wheel & Fuel in FC mode & Battery in FE mode \\
\hline
\hline       
front vehicle & 100.0 & 100.0 & 100.0 \\
$N_T = 0$ & 90.6 & 91.9 & 90.8\\
$N_T = 3$ & 87.6 & 91.4 & 89.6 \\
$N_T = 8 $ & 85.2 & 90.3 & 89.5
\end{tabular}
\caption{Effects of trust horizon.}
\label{tab:Nt-effect}
\end{subtable}
\end{table}

\subsection{Effects of reduced measurement delay}
\label{sec:red_measure_delay}
In CAVs, the states of the other CAVs are accessible through the on-board sensors and V2V communication.
With smaller delay in the measurement, the current state \eqref{eq:rob_state_pred} can be estimated less conservatively and more accurately. 

Figure~\ref{fig:clplot:4-1} depicts the trajectories of the front and the controlled vehicles, in terms of their inter-vehicular distance, velocities, and wheel torque when the controlled vehicle has the front vehicle measurement with a known delay of \SI{0}{\s}, \SI{0.2}{\s}, and \SI{0.4}{\s}, i.e. $h=0,1,2$, respectively.
With smaller delay, the controlled vehicle decelerates the latest and catches up with the front vehicle more quickly and reaches a smaller constant distance gap ($t \geq \SI{30}{\s}$); this is clearly seen in the zoomed section in the plots. 

One can notice, in Table~\ref{tab:delay-effect}, that a smaller delay decreases wheel energy consumption of the controlled vehicle down to $90.6\%$.
This amount of wheel energy saving is approximately $9\%$ battery energy saving or $8\%$ fuel energy saving.

\begin{figure}
\centering
\includegraphics[width=\columnwidth]{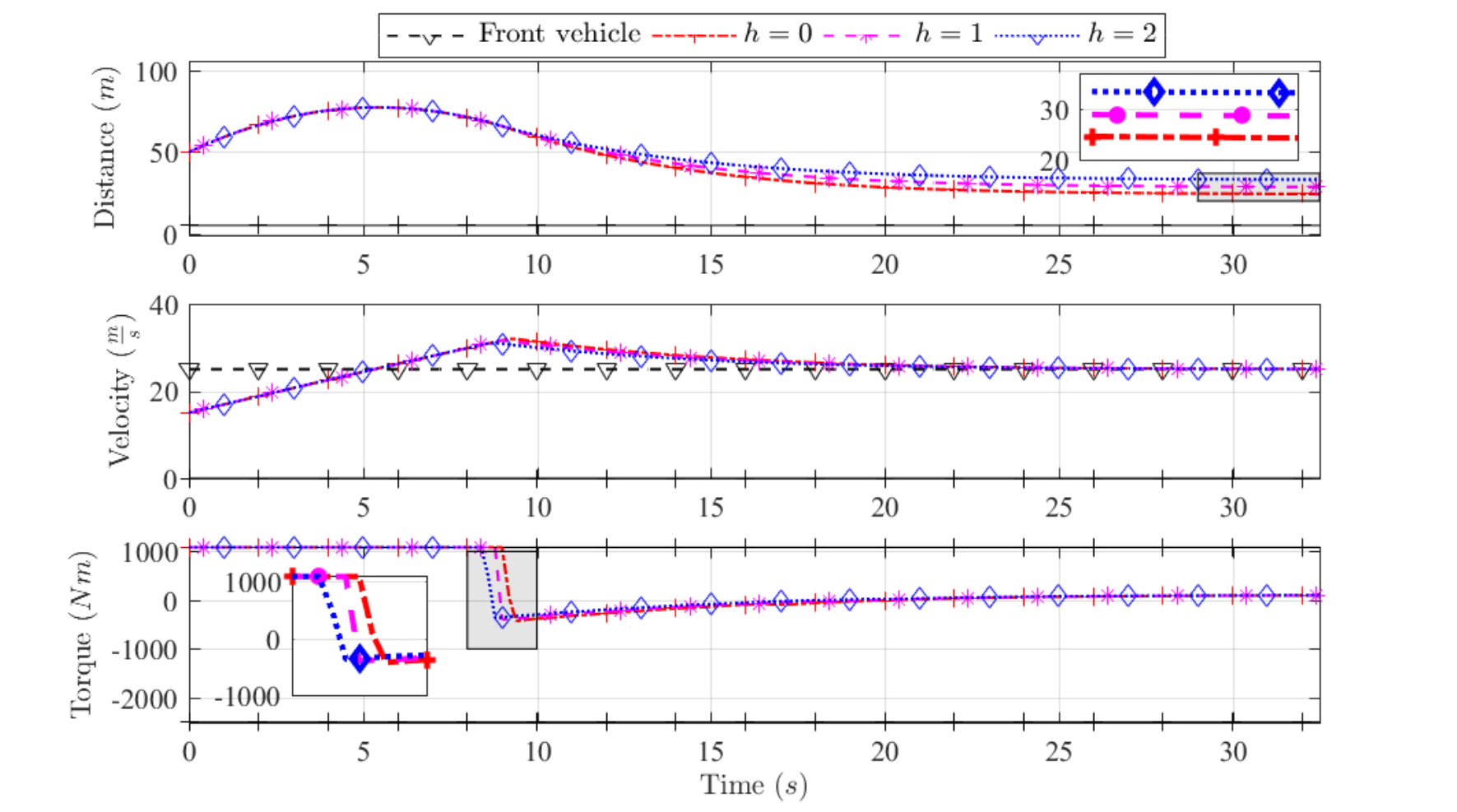}
\caption{Plots of the distance between the controlled and the front vehicles, their velocities, and the wheel torques of the controlled vehicle at different measurement delay $h$.} \label{fig:clplot:4-1}
\end{figure}

\subsection{Effects of the front vehicle's maximum deceleration}
\label{sec:small_dec}
When there is no communication between the two vehicles, the following vehicle may assume that the front vehicle can perform maximum braking (at its physical limit) at any time.
Unfortunately, this assumption causes a conservative behavior to the following vehicle i.e. a larger inter-vehicular distance gap.
This can be mitigated if the following vehicle receives the limit of the front vehicle's maximum braking in their V2V message as shown in equation \eqref{eq:v2v_msg}.

This section shows how the magnitude of the front vehicle's maximum deceleration can affect the performance of the controlled vehicle.
First, it affects the shape of a robust control invariant set used in \eqref{eq:term_cons}. Figure~\ref{fig:set_a_min} shows the robust control invariant sets obtained using the method from Appendix for different magnitudes of the front vehicle's maximum deceleration;  
the smaller the deceleration magnitude, the bigger the robust control invariant set.
Second, with a smaller magnitude of the front vehicle's maximum deceleration, the front vehicle acceleration predictor in \eqref{eq:front_a_def} is less conservative and produces a velocity trajectory with less braking.
As a result of these two effects, the magnitude of the front vehicle's maximum deceleration influences the performance of CACC \eqref{eq:mpcproblem}-\eqref{eq:optimal_input}.

Figure~\ref{fig:clplot:4-2} depicts the trajectories of the front and the controlled vehicles, in terms of their inter-vehicular distance, velocities, and wheel torque when the front vehicle sends to the controlled vehicle its maximum braking rates of $-9\si{m \per \square \s}$, $-6\si{m \per \square \s}$, and $-3\si{m \per \square \s}$. 
The smaller the front vehicle's maximum deceleration magnitude, the faster the controlled vehicle reaches a constant velocity and distance. Moreover, the controlled vehicle's constant distance gap is smaller with a smaller front vehicle maximum deceleration magnitude.
As a result of a shorter constant distance gap, the smaller the front vehicle's maximum deceleration assumption is, the more saving on wheel energy, fuel, or battery the controlled vehicle achieves during the steady state.

\begin{remark}
Another benefit of knowing maximum acceleration and deceleration of the preceding vehicle in advance is that this information can be utilized to guarantee the string stability of our CACC controller as done in \cite[Eq. (25)]{kianfar2011receding}. 
\end{remark}

\begin{figure}
\centering
\includegraphics[width=\columnwidth]{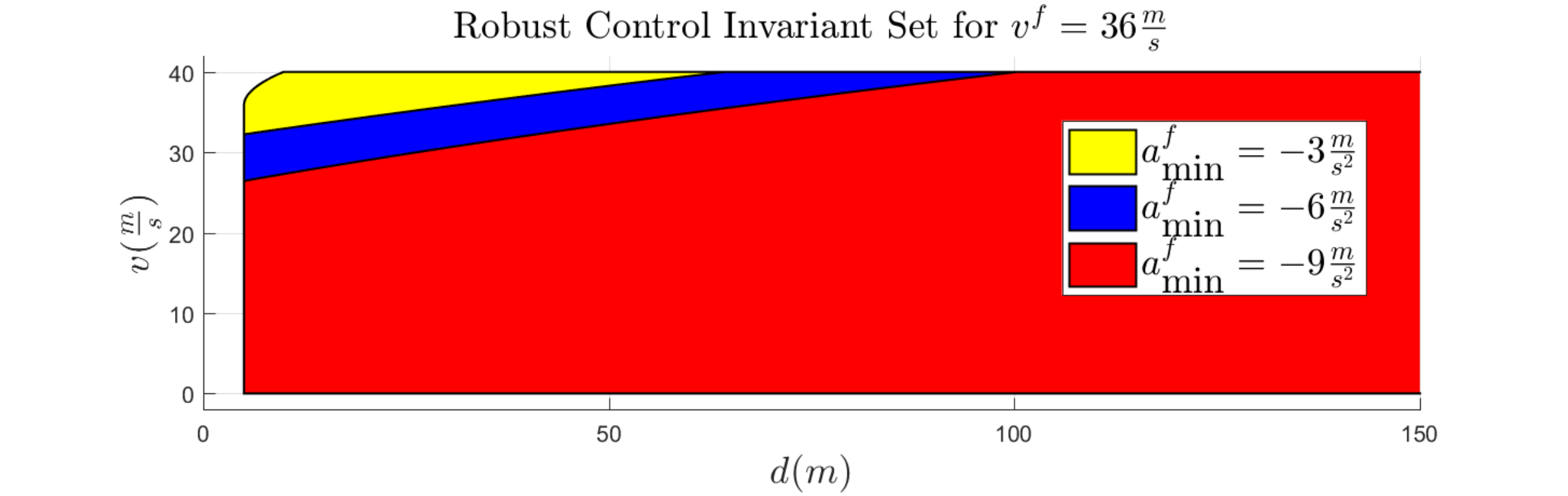}
\caption{Robust control invariant set for different front vehicle's maximum deceleration rates} \label{fig:set_a_min}
\end{figure}

\begin{figure}
\centering
\includegraphics[width=\columnwidth]{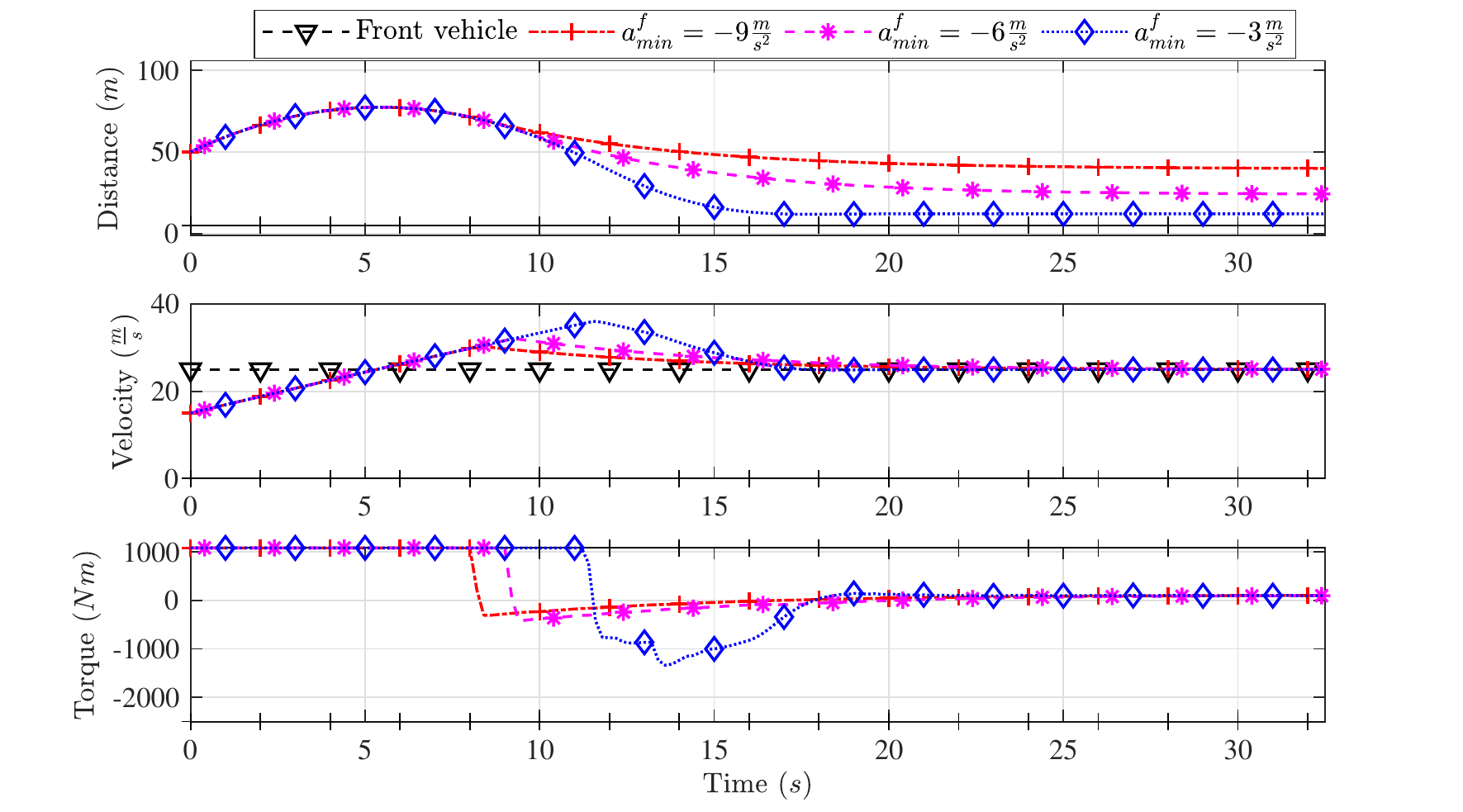}
\caption{Plots of the distance between the controlled and the front vehicles, their velocities, and the wheel torques of the controlled vehicle at the front vehicle's different maximum braking rates $a^f_{\textnormal{min}}$.}\label{fig:clplot:4-2}
\end{figure}

\subsection{Effects of trusting acceleration forecast}
\label{sec:trust_vel_forecast}
This section is devoted to the scenario where the front vehicle shares its acceleration forecasts with $N_T > 0$.
Similar approach and analysis were conducted in \cite{smith2019balancing} for the purpose of maximizing road throughput. 
It is reported in \cite{smith2019balancing,kianfar2013distributed} that the string stability is achieved when the trust horizon $N_T$ is long enough.
In this simulation, we focus on the energy saving effect of a length of a trust horizon $N_T$ in the V2V message \eqref{eq:v2v_msg}.

Figure~\ref{fig:clplot:4-3} depicts the trajectories of the front and the controlled vehicles, in terms of their inter-vehicular distance, velocities, and wheel torque when the front vehicle sends to the controlled vehicle the acceleration forecast with a trust horizon of $0$, $3$, and $8$ steps (0s, 0.6s, and 1.6s, respectively). 
A longer trust horizon has a similar effect to a smaller maximum deceleration magnitude of the front vehicle.
In short, with a longer trust horizon, the controlled vehicle reaches a constant distance gap and velocity faster and this constant gap is smaller.
As seen in Table~\ref{tab:Nt-effect}, The longer the trust horizon is, the more energy saving on wheel energy, fuel, and battery the controlled vehicle achieves during the steady state.



\begin{figure}
\centering
\includegraphics[width=\columnwidth]{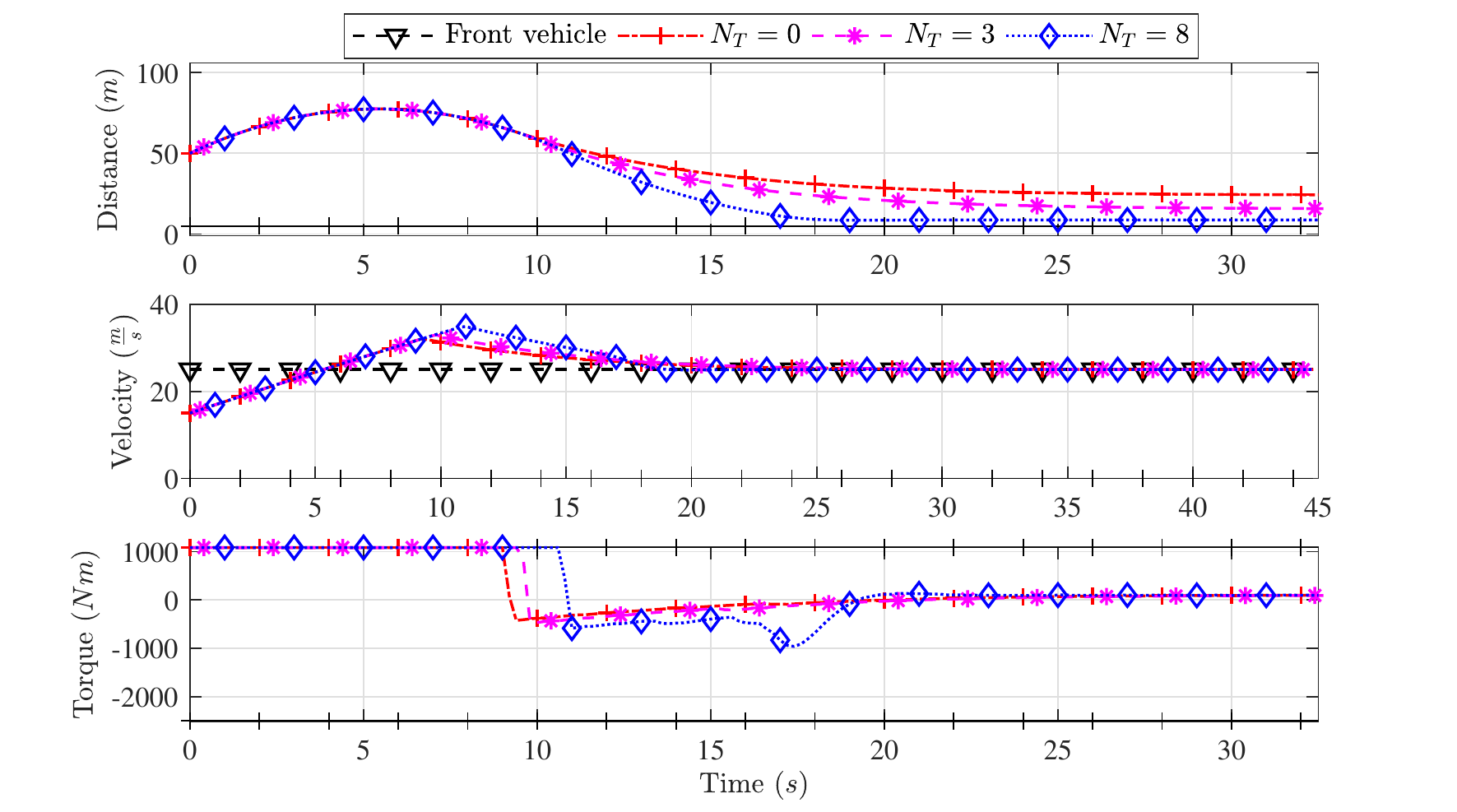}
\caption{Plots of the distance between the controlled and the front vehicles, their velocities, and the wheel torques of the controlled vehicle with different trust horizon $N_t$.}\label{fig:clplot:4-3}
\end{figure}

\subsection{Effects of the front vehicle's distance/velocity measurement noise magnitude} \label{sec:small_vel_uncer}
Communication improves the accuracy of the environment perception by allowing the direct exchange of states among vehicles. 
With only a radar, noise magnitudes can be significant \cite{floudas2005survey}.
In this section we study the scenario where the V2V message from the front vehicle offers measurement of the front vehicle velocity and location with lower measurement noises. 
To simulate the measurement noise, we generate uniformly distributed random numbers, of which magnitudes are bounded by $n_{\textnormal{max}}$, and add them to the distance and front vehicle velocity measurements in \eqref{eq:rob_ob_init}.

Figure~\ref{fig:clplot:4-4} shows the trajectories of the front and the controlled vehicles, in terms of their inter-vehicular distance, velocities, and wheel torque when the controlled vehicle measures the front vehicle distance with noise magnitudes bounded by $0 \si{m}$, $0.15\si{m}$, and $0.3\si{m}$ and the front vehicle velocity with noise magnitudes bounded by $0 \si{m \per \s}$, $0.15 \si{m \per \s}$, and $0.3\si{m \per \s}$.
One can observe, as expected, that the bigger noise bounds result in larger fluctuations in the wheel torque trajectories.

\begin{figure}
\centering
\includegraphics[width=\columnwidth]{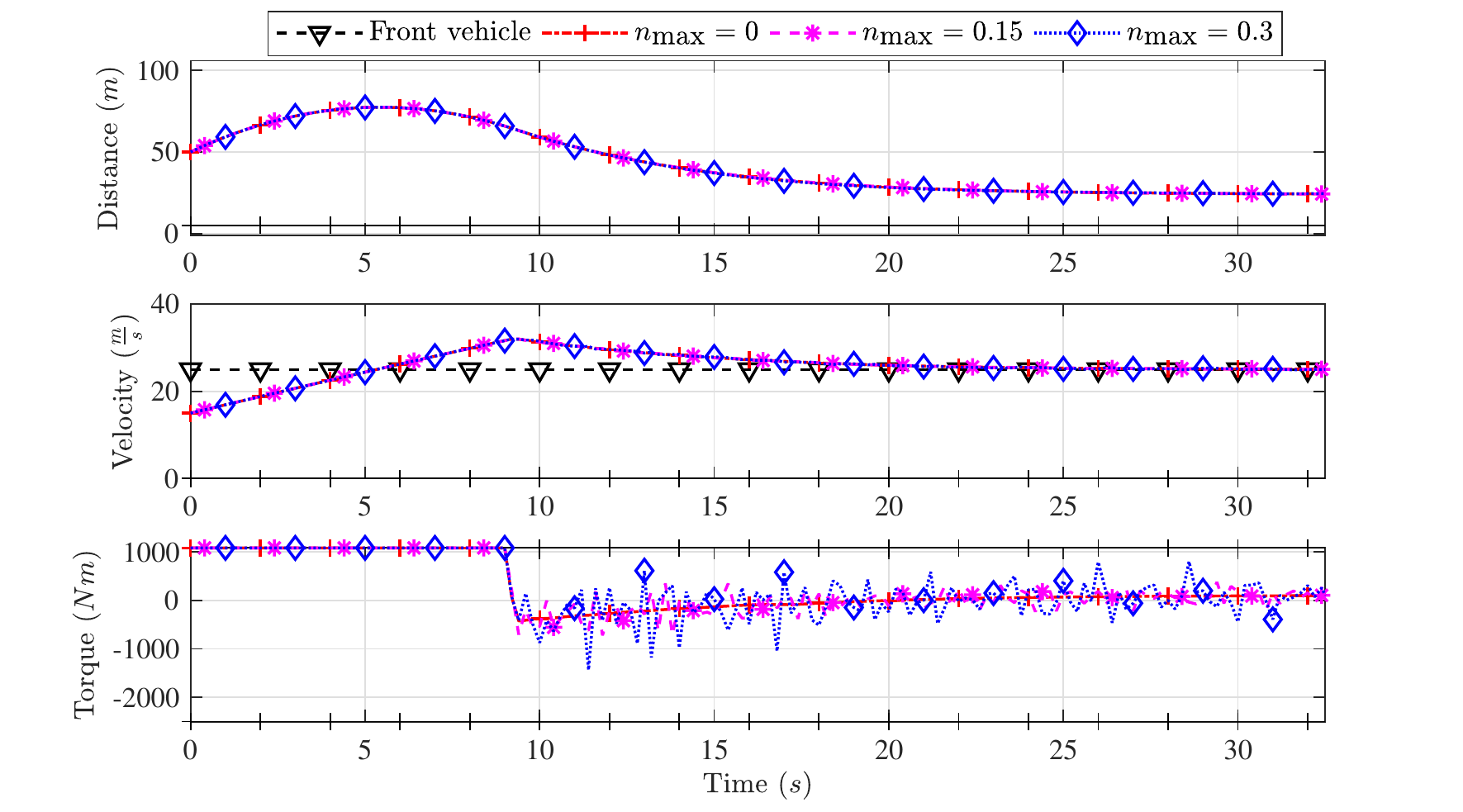}
\caption{Plots of the distance between the controlled and the front vehicles, their velocities, and the wheel torques of the controlled vehicle with different maximum magnitudes of the front vehicle velocity and distance $n_{\textnormal{max}}$.}\label{fig:clplot:4-4}
\end{figure}

\section{Conclusion}
\label{sec:conclusion}

In this paper, we first fit a well known longitudinal vehicle dynamics model with experimentally identified parameters for compact vehicles.
Then, we investigate the effects of smaller vehicular gap on both air drag and energy consumption.
Motivated by the experimental results on the effects of a vehicular gap on energy saving, we propose a car-following cooperative adaptive cruise controller which seeks to minimize a vehicular gap exploiting the V2V communication.
We show, under some mild assumptions, that the proposed control system can robustly guarantee the safety of the controlled vehicle.
Then, simulation results are presented to show the effects of V2V connectivity with various message contents.
Our future work will be articulated in three directions. 
First, the presented control design only considers calculating control action with V2V communication with a-priori known contents. 
Our next work will study interaction between control and V2V information which needs to be transmitted concurrently.
Second, string stability analysis given different V2V information will be conducted by simulation. 
Third, experimental validation of our proposed controller in closed loop may be carried out. 
In the experiments, the lateral dynamics can be taken into account in the control design as the lateral distance gap can also play important role in energy saving effects of the compact platoon \cite{mcauliffe2018wind}.

\section*{Acknowledgement}

The information, data, or work presented herein was funded in part by the Advanced Research Projects Agency-Energy (ARPA-E), U.S. Department of Energy, under Award Number DE-AR0000791. The views and opinions of authors expressed herein do not necessarily state or reflect those of the United States Government or any agency thereof.

\appendix[Safety Guarantee of our Control System]\label{sec:safety_guarantee}
In the conflated control design such as the reinforcement learning method \cite{kreidieh2018dissipating}, the safety guarantee (to satisfy system constraints at all times) is very difficult to ensure. 
In robust MPC, the safety guarantee can be systemically achieved by enforcing the state at the terminal horizon to be in the robust control invariant set (see \cite{mayne2000constrained} for details).

To compute a robust control invariant set for system \eqref{eq:long_dyn} subject to constraints \eqref{eq:stateinputconstraints}, we build on the method from \cite{lefevre2016learning}, which is suited to linear hybrid systems.

First we introduce the uncertain linear hybrid system
\begin{equation}
\label{eq:lin_long_dyn}
x^l(k+1)= 
\begin{bmatrix}
d^l(k+1)\\
v^l(k+1)\\
v^{f,l}(k+1)
\end{bmatrix} =
\begin{bmatrix}
d^l(k) + t_s(v^{f,l}(k) - v^l(k)) \\
v^l(k) + \frac{t_s}{m} \Big(\frac{u^l(k)}{R_w} \Big) \\
\big( v^{f,l}(k) + t_s a^{f,l}(t)\big)^+
\end{bmatrix}
\end{equation}
and the state/input constraints and the uncertainty set are compactly written as
\begin{equation}\label{eq:lin_stateinputconstraints}
x^l(k) \in \mathbb{X}, \, u^l(k) \in \mathbb{U}^l,\, a^{f,l}(k) \in \mathbb{A}^f, 
\end{equation}
where $\mathbb{X}$ and $\mathbb{A}^f$ are defined in \eqref{eq:stateconstraints} and \eqref{eq:a_front_set}, respectively, and
\begin{equation} \label{eq:Ul}
\begin{aligned}
\mathbb{U}^l:= \{u^l(k):  &T^w_{\textnormal{min}} - R_w mg \munderbar{C}_r \leq u^l(k), \nonumber \\
& u^l(k) \leq T^w_{\textnormal{max}} - R_w (mg\bar{C}_r +\rho A \bar{C}_x v_{\textnormal{max}}^2) \}.
\end{aligned}
\end{equation}
In \eqref{eq:Ul}, $\munderbar{C}_r$ and $\bar{C}_r$ are the minimum and the maximum values of $( \sin(\vartheta) + C_r\cos(\vartheta))$ and $\bar{C}_x$ is the maximum values of ${C}_x(\cdot)$. 

Using the method from \cite{lefevre2016learning}, we can obtain a polytopic representation of a robust control invariant set for the system \eqref{eq:lin_long_dyn} with the constraints \eqref{eq:lin_stateinputconstraints}. More specifically, we computed the robust control invariant sets for the system \eqref{eq:lin_long_dyn} with the constraints \eqref{eq:lin_stateinputconstraints} assuming different initial front vehicle velocities discretized from the complete stop to the maximum vehicle velocity.
The collection of these invariant sets then represents the robust control invariant set for our system. It is also noted that to compute each invariant set, we used the Multi-Parametric Toolbox (MPT) to compute the backward reachability set of the linear system until the set converges \cite{kvasnica2004multi}. 

Then, we are in place to state the following theorem.
\begin{theorem}
\label{inv_set}
A set $\mathbb{C}$ is a robust control invariant set for the system \eqref{eq:cont_oriented_long_dyn}-\eqref{eq:a_front_set} if $\mathbb{C}$ is also a robust control invariant set for the linear system \eqref{eq:lin_long_dyn}-\eqref{eq:lin_stateinputconstraints}.
\end{theorem} 
\begin{proof}
Proof for Theorem \ref{inv_set} is straightforward from the definition of a robust control invariant set \cite{blanchini1999set}.
A robust control invariant set $\mathbb{C}$ for the system \eqref{eq:lin_long_dyn} with constraints \eqref{eq:lin_stateinputconstraints} means that if $x^l_0 \in \mathbb{C}$, there exists a control law $u^l(t) = \Phi(x^l(t)) \in \mathbb{U}^l$ such that $x^l(t) \in \mathbb{C}$, for all $a^{f,l}(t) \in \mathbb{A}^f$.

Now, observe that the control policy $u(t) = \Phi(x(t)) + R_w mg ( \sin(\vartheta) + C_r\cos(\vartheta)) + R_w\rho A {C}_x(t) v(t)^2$ makes the set $\mathbb{C}$ also a robust control invariant set for our original system \eqref{eq:cont_oriented_long_dyn}-\eqref{eq:a_front_set} as it cancels out the nonlinear terms in \eqref{eq:cont_oriented_long_dyn}, $mg(\sin{\vartheta(t)} + {C}_r\cos{\vartheta(t)}) - \frac{1}{2} \rho A C_x(d(k |t))  v(k |t)^2 $.
Moreover, it is easy to see that the control policy $u(t)$ always satisfies the input constraints \eqref{eq:inputconstraints}.
\end{proof}

It is also noted that the defined robust control invariant set $\mathbb{C}$ can be applied to the control designs other than the robust MPC. For instance, in reinforcement learning control design (as suggested in \cite{kreidieh2018dissipating}), a risk-sensitive reinforcement learning against the probability of the state to be outside $\mathbb{C}$ can be used to train the safe CACC control policy.

\begin{remark}
\label{rem2}
The robust control invariant set $\mathbb{C}$ is unbounded both in the positive distance direction and the positive front vehicle velocity direction.
\end{remark}


Now we are in place to state the following result.
\begin{theorem}
\label{recur_theo}
Consider the RMPC problem \eqref{eq:mpcproblem}-\eqref{eq:optimal_input}  ($N_p\geq1$) for the system \eqref{eq:cont_oriented_long_dyn}-\eqref{eq:a_front_set}. Assume that: (i) for $t \leq 0$, the RMPC problem is feasible, (ii) the closed loop state feedback is provided by the control oriented model \eqref{eq:cont_oriented_long_dyn} on a flat road ($\vartheta(t)=0 \, \forall t$), and (iii) state measurements are accurate ( $n_{d}^{\textnormal{max}} = n_{v^f}^{\textnormal{max}} = 0$) but the front vehicle distance and velocity are subject to delay. Then, the RMPC problem \eqref{eq:mpcproblem}-\eqref{eq:optimal_input} is persistently feasible and the closed loop system never violates the constraints \eqref{eq:stateinputconstraints}.
\end{theorem} 
The proof the Theorem~\ref{recur_theo} is backed by the property of the RMPC. First, it is trivial that when the initial state is measured without delay (i.e., $h=  0$ and $\Tilde{x}(t|t) =\bar{x}(t)$), the RMPC problem is persistently feasible due to the property of the robust control invariant set $\mathbb{C}$ \cite{mayne2000constrained, lefevre2016learning}, and the robust front vehicle acceleration forecast \eqref{eq:front_a_def}. Moreover, this guaranteed feasible input trajectory always satisfies every state, input, and terminal constraint in \eqref{eq:mpcproblem} even when the initial state, $\Tilde{x}(t|t)$, is the delayed yet shifted measurement as done in \eqref{eq:rob_state_pred}.

\bibliography{bibliography}

\end{document}